\documentclass[11pt,a4paper]{article}
\usepackage{amssymb,slashed,mathrsfs,amsmath,marvosym,amsthm,mdwlist,mathdots,dsfont,color}

\usepackage[colorlinks]{hyperref}
\definecolor{IndianRed}{rgb}{0.8,0.36,0.36}
\definecolor{DarkGreen}{rgb}{0,0.5,0}
\definecolor{DarkBlue}{rgb}{0,0,0.5}

\newcommand{\be}{\begin{equation}}
\newcommand{\ee}{\end{equation}}

\newcommand{\nn}{\nonumber}
\newcommand{\dt}[2][L]{\det\nolimits_{#1}\left[ #2 \right ]}

\newcommand{\de}[1]{\frac{\partial}{\partial #1}}
\newcommand{\stack}[2]{\begin{subarray}{c} #1 \\ #2 \end{subarray}}

\numberwithin{equation}{section}
\theoremstyle{plain}
\newtheorem{thm}{Theorem}[section]
\newtheorem{lm}[thm]{Lemma}
\newtheorem{prop}[thm]{Proposition}
\newtheorem*{cor}{Corollary}

\theoremstyle{definition}
\newtheorem{defn}{Definition}[section]

\theoremstyle{remark}

\newcommand{\lmref}[1]{Lemma~\ref{lm:#1}}
\newcommand{\propref}[1]{Proposition~\ref{prop:#1}}

\newcommand{\secref}[1]{Section~\ref{sec:#1}}
\newcommand{\apref}[1]{Appendix~\ref{ap:#1}}
\newcommand{\defref}[1]{Definition~\ref{def:#1}}

\newcommand{\dd}{{\rm d}\hspace{1pt}}
\newcommand{\e}{{\rm e}\hspace{1pt}}

\begin{document}
\hypersetup{%
    urlcolor=DarkBlue, 
    linkcolor=IndianRed, 
    citecolor=DarkGreen 
} 
\title{Separation of variables for symplectic characters}

\author{Jan de Gier\footnote{jdgier@unimelb.edu.au}~ and Anita Ponsaing\footnote{a.ponsaing@ms.unimelb.edu.au}
\bigskip\\
{\small\em
\begin{minipage}{0.9\textwidth}
\begin{center}
Department of Mathematics and Statistics, The University of Melbourne,VIC 3010, Australia
\end{center}
\end{minipage}
}
}

\date{\today}

\maketitle

\begin{abstract}
We perform separation of variables for the symplectic Weyl character using Sklyanin's scheme. Viewing the characters as eigenfunctions of a quantum integrable system, we explicitly construct the separating operator using the Q-operator method. We also construct the inverse of the separating operator, as well as the factorised Hamiltonian.
\end{abstract}

\section{Introduction}
Symplectic characters, or Schur polynomials for the root system of type $C$, play an important role in the representation theory of finite groups \cite{FultonH91}. They are also used as generating functions for counting problems with boundaries in enumerative combinatorics, see e.g. \cite{CiucuK09,Okada06}. In recent years symplectic characters have appeared in square lattice critical bond percolation on lattice strips with boundaries \cite{DFZJ07,ZJ07,dGPS09,Cantini09}, with applications to the spin quantum Hall effect \cite{dGNP10}. In these cases, they arise from polynomial solutions of the $q$-deformed Knizhnik-Zamolodchikov equation \cite{FrenR92,Smirn92,KasaT07,Kasa08} where $q$ is the third root of unity, and appear in the exact finite size expressions for normalisations and correlation functions.
 
For these applications, one naturally would like to understand the asymptotic behaviour of the symplectic character as the number of variables tends to infinity, and in particular in the limit when all but a finite number of variables are set to $1$. Okounkov and Olshanki studied this limit for both type $A$ Jack polynomials and type $BC$ orthogonal polynomials \cite{OkouO98,OkouO06}. For both of these problems, the authors considered polynomials with a degree growing linearly with the number of variables. In contrast, the symplectic characters which arise in the aforementioned bond percolation models have a degree that is quadratic in the number of variables and, as far as we are aware, their asymptotic behaviour is an open problem. 

In this paper we study a method for the separation of variables (SoV) of symplectic characters. Our underlying motivation for studying this problem is to understand the asymptotic limit as described above. The aim of SoV is to produce a product of factors, one in each variable. Usingthis, the problem of the asymptotics will be reduced to that of finding the asymptotics for each factorised part. In order for this approach to be useful the SoV method must also be invertible, which is one of the main technical hurdles in SoV. 

Kuznetsov and Sklyanin \cite{KuzSkl06} described a method of SoV for symmetric functions, based on earlier work with Mangazeev on Jack polynomials of type $A$ in \cite{KuzMS03}. The method used by these authors is based on the $Q$-operator formalism for SoV initiated by Kuznetsov and Sklyanin \cite{KuzSkl98}. The $Q$-operator was first introduced by Baxter \cite{Bax82}, in his solution of the $8$-vertex model. 

Symmetric polynomials are eigenfunctions of certain multivariate differential operators, or Hamiltonians $\{H_i\}$, which form a quantum integrable system. The $Q$-operator for a quantum integrable system $\{H_i\}$ is a quantisation of the B\"acklund transformation for the corresponding classical integrable system. This connection was first found by Pasquier and Gaudin \cite{PasqG92}, who discovered a correspondence between the classical B\"acklund transformation and the $Q$-operator for the periodic Toda lattice. The operator $Q$ can be realised as an integral operator with a simple kernel. In the quasi-classical limit, this kernel turns into the generating function of the canonical transform \cite{KuzSkl98,PasqG92}.

A number of examples are given in \cite{KuzSkl06} for the application of the $Q$-operator method to systems of symmetric polynomials associated to the root system of type $A$. In particular, \cite{KuzSkl06} discusses SoV for Schur functions, being a special case of Jack polynomials. As indicated above, this paper aims to extend SoV using the $Q$-operator method to the irreducible character of the symplectic group, $\chi_\lambda$, which is the Schur polynomial for the root system of type $C$. SoV methods for root systems other than type $A$ have not been studied much in previous literature, the only case which we are aware of is for the open quantum $\rm{sl}(2,\mathbb{R})$ spin chain \cite{DerkachovKM03}. 

The key ingredient of the $Q$-operator method is the construction of the separating operator $\mathcal S$ whose action on a polynomial $P_\lambda$ is proportional to a product of single-variable polynomials,
\be
\label{eq:SPq}
\big(\mathcal S P_\lambda\big)(x_1,\ldots ,x_L) =P_\lambda(1,\ldots,1)\prod_{i=1}^L q_\lambda (x_i).
\ee
For the case of the symplectic character, as with Schur functions, $\mathcal S$ is invertible. The structure of the method is given in the next section.

\subsection{Notation}
Throughout the paper, $L$ is the number of variables. We use $\lambda$ to refer to an arbitrary partition of length $L$; i.e.\ an $L$-tuple of positive integers $(\lambda_1,\ldots,\lambda_L)$, which has the property that $\lambda_i\geq\lambda_{i+1}$. We also define the partition $\delta=(L,\ldots ,L-i+1,\ldots ,2,1)$, and use $\mu$ to denote the partition $\lambda+\delta$, so $\mu_i=\lambda_i+L-i+1$.

The symbol $\dt[k]{a}$ refers to the determinant of the matrix $a$ which is $k\times k$ in size. A bold letter such as $\mathbf{x}$ refers to the list of variables $x_1,\ldots,x_L$.

The usual symmetric group of size $L$ is given by $\mathcal S_L$, and the Weyl group for the type $B_L$ root system is given by $\mathcal W_L$.

\section{The $Q$-operator method}
\label{sec:method}
We consider a quantum integrable system with a commuting set of Hamiltonians, given by differential operators $H_i$ ($i\in\{1,\ldots,L\}$). These have eigenstates equal to the polynomial $P_\lambda$:
\be
\label{eq:hamil}
H_j P_\lambda (x_1,\ldots ,x_L)=h_j(\lambda )P_\lambda (x_1,\ldots ,x_L),
\ee
where the $h_j$ are the corresponding eigenvalues which depend on the multi-degree $\lambda$ of the polynomial.

The aim of the $Q$-operator method is to find a version of this spectral problem that involves a related factorised polynomial in place of $P_\lambda$. Acting on both sides of \eqref{eq:hamil} with $\mathcal S$, we have
\[ \mathcal SH_j\mathcal S^{-1}\ \big(\mathcal SP_\lambda\big) (x_1,\ldots ,x_L)=h_j(\lambda )\ \left(\mathcal SP_\lambda\right) (x_1,\ldots ,x_L),\]
and using $\eqref{eq:SPq}$,
\[ \left (\mathcal SH_j\mathcal S^{-1}\right ) \prod_{i=1}^L q_\lambda (x_i)=h_j(\lambda )\prod_{i=1}^L q_\lambda (x_i).\]

The main results of this paper are the explicit construction of the operators $\mathcal S$ and $\mathcal S^{-1}$ for the symplectic characters, as well as a ``factorised Hamiltonian'' $\widetilde H_j$, which acts in the same way as $\mathcal SH_j\mathcal S^{-1}$ on the factorised polynomial. $\widetilde H_j$ is not uniquely defined, and can be constructed using $q_\lambda$ as shown in the next section.

\subsection{Factorised Hamiltonian}
\label{sec:fachamil}
The operator $\widetilde H_j$ can be constructed from a single variable operator in the following way. If there exists a differential operator in $x_i$, denoted $W_i$, such that
\be
\label{eq:Wq0}
W_i q_\lambda (x_i)=0,\qquad [W_i\ , q_\lambda (x_k)]=0\quad (k\neq i),
\ee
then we construct the operator
\[ W_{i,j}=W_i+h_j(\lambda ) .\]
Now, noting that $W_{i,j}$ commutes with any function of $x_k$ where $k\neq i$, it is easy to see that
\be
\label{eq:Wqhq}
W_{i,j} \prod_{k=1}^L q_\lambda (x_k)=h_j(\lambda )\prod_{k=1}^L q_\lambda (x_k),\qquad \forall i,
\ee
and therefore any linear combination of the form
\[\widetilde H_j=\sum_{i=1}^L c_i W_{i,j},\qquad\qquad \sum_{i=1}^L c_i=1,\]
will also satisfy \eqref{eq:Wqhq}.

\subsection{Factorisation of the separating operator $\mathcal S$}
We will show that $\mathcal S$ can be written as a product of operators,
\be
\label{eq:SQQ}
\mathcal S= \big(\rho _0Q_{z_1}\ldots Q_{z_L}\big)\big|_{z_1=x_1,\ldots,z_L=x_L},
\ee
where the operator $Q_{z_i}$ is an integral operator that acts as
\be
\label{eq:QPqP}
\big(Q_{z_i}P_\lambda\big)(x_1,\ldots ,x_L)=q_\lambda (z_i) P_\lambda (x_1,\ldots ,x_L),
\ee
and $\rho_0$ sends $f(x_1,\ldots,x_L)\rightarrow f(1,\ldots,1)$. Furthermore, the $Q_{z_i}$ have the important properties
\[ \left[Q_{z_i},Q_{z_j}\right ]=0,\qquad \left[Q_{z_i},H_j\right ]=0,\qquad\forall i,j.\]
Having found the operator $Q_{z_i}$, we introduce new operators $\mathcal A_i$ for which
\[ \big(\rho_{i-1}Q_{z_i}\big)\big|_{z_i=x_i}=\mathcal A_i \rho _i.\]
Here, $\rho_j$ sends $f(x_1,\ldots,x_L)\rightarrow f(x_1,\ldots,x_j,1,\ldots,1)$. These new operators act as
\be
\label{eq:APqP}
\big(\mathcal A_iP_\lambda\big) (x_1,\ldots ,x_i,1,\ldots ,1)=q_\lambda (x_i)P_\lambda (x_1,\ldots ,x_{i-1},1\ldots ,1),
\ee
and \eqref{eq:SQQ} becomes the factorisation
\[ \mathcal S=\mathcal A_1\ldots\mathcal A_L.\]

\subsection{Summary}
Given a particular family of polynomials, the $Q$-operator method can be condensed into five steps:
\vspace{.3cm}\\
\begin{minipage}{\textwidth}
\begin{enumerate*}
\item{Specify $H_j$ and $h_j(\lambda)$ in \eqref{eq:hamil},}
\item{Specify $W_i$, verify \eqref{eq:Wq0}, and construct the factorised Hamiltonian,}
\item{Construct $\mathcal S^{-1}$,}
\item{Find $Q_z$ such that \eqref{eq:QPqP} is satisfied,}
\item{Find $\mathcal A_i$ such that \eqref{eq:APqP} is satisfied, and construct $\mathcal S$.}
\end{enumerate*}
\end{minipage}\vspace{.3cm}

The inverse separating operator is given before the construction of $\mathcal S$ for two reasons; firstly, for the symplectic character it turns out to be simpler in form, and secondly, the details relate to that of the differential operators $H$ and $W$ used in the first two steps.

As previously stated, this paper focuses on the case $P_\lambda=\chi_\lambda$, the symplectic character which will be defined in \secref{symp}. For this case, we closely follow \cite{KuzSkl06}, in which the factorisation of the Schur functions is treated. The technical detail in the symplectic case however is more elaborate.

The layout of this paper is as follows: \secref{symp} describes the symplectic character and some of its properties. The sections following proceed in the order of the steps above. \secref{invS} also describes the inverse of the operator $\mathcal S_k=\mathcal A_1\ldots\mathcal A_k$, which satisfies
\[ \mathcal S_k^{-1}\prod_{i=1}^k q_\lambda (x_i)=\frac{\chi_\lambda (x_1,\ldots ,x_k,1,\ldots ,1)}{\chi_\lambda (1,\ldots ,1)}.\]

\section{Properties of the symplectic character}
\label{sec:symp}

The polynomials we will consider are the irreducible characters of the symplectic group. These form a basis of $\mathbb C[x_1^{\pm},\ldots,x_L^{\pm}]^{\mathcal W_L}$, the ring of Laurent polynomials with complex coefficients which are symmetric in $x_i$ and invariant under $x_i\rightarrow x_i^{-1}$. Each character is labelled by a partition $\lambda=(\lambda_1,\ldots,\lambda_L)$. We refer to them as symplectic characters \cite{Weyl39,FultonH91}, and define them as follows,
\begin{defn}
\[ \chi _\lambda (x_1,\ldots ,x_L)=\frac{a_\mu (x_1,\ldots ,x_L)}{a_\delta (x_1,\ldots ,x_L)},\]
where
\be
\label{eq:a}
a_\mu (x_1,\ldots ,x_L)=\dt{x_i^{\mu _j}-x_i^{-\mu _j}},
\ee
and $\mu=\lambda+\delta$.
\end{defn}
The polynomials $a_\mu$ are Laurent polynomials in the ring $\mathbb C[x_1^{\pm},\ldots,x_L^{\pm}]$ which are antisymmetric in $(x_1,\ldots,x_L)$ and under $x_i\rightarrow x_i^{-1}$. These antisymmetries mean that $a_\mu(\mathbf{x})=0$ whenever $x_i=x_j$ or $x_i=1$. For the definition of $\mathcal A_i$, we will need an expression for $\chi_\lambda$ when some of the arguments are set to 1, so we must find an alternate definition for $\chi_\lambda$ in this limit. We therefore define a `truncated' version of $a_\mu$:
\begin{defn}
\be
\label{eq:atrunc}
a^{(k)}_\mu (x_1,\ldots ,x_{k-1})=\dt{\begin{array}{c}
\left[ x_i^{\mu _j}-x_i^{-\mu _j} \right ]_{i < k} \\[5mm]
\left[\mu _j^{2(L-i)+1}\right ]_{i\geq k}
\end{array}}.
\ee
\end{defn}

\begin{lm}
\label{lm:kprop}
For $i=k,\ldots,L$, set $x_i=\e^{\varepsilon u_i}$ with $u_i\in\mathbb R$. Then
\[
a^{(k)}_\mu(x_1,\ldots,x_{k-1}) \sim c_k(\varepsilon) a_\mu(x_1,\ldots,x_L)
\]
in the limit as $\varepsilon\rightarrow 0$. The prefactor $c_k(\varepsilon)$ does not depend on $\mu$.
\end{lm}

\begin{proof}
We first prove the above for the case where $k=L$. In the the $L$th row of the determinant $a_\mu(x_1,\ldots,x_{L-1},\e^{\varepsilon})$, the $j$th element is
\begin{align*}
\e^{\varepsilon \mu _j}-\e^{-\varepsilon \mu _j}&=2\sinh (\varepsilon \mu _j)\\
&=2\varepsilon \mu _j+ {\rm h.o.t.}
\end{align*}
Taking this to first order in $\varepsilon$, we can factor $2\varepsilon $ out of the determinant, leaving $\mu_j$ in the bottom row. Then
\[a^{(L)}_\mu(x_1,\ldots,x_{L-1})=\lim_{\varepsilon\rightarrow 0}\frac1{2\varepsilon }a_\mu(x_1,\ldots,x_{L-1},e^\varepsilon),\]
so $c_L(\varepsilon)=1/2\varepsilon $.

In the case of general $k$, the parameters $u_i$ allow us to take multiple arguments to $1$ along distinct trajectories. In the rows from $k$ to $L$ of $a_\mu(x_1,\ldots,x_{k-1},\e^{\varepsilon u_k},\ldots,\e^{\varepsilon u_L})$, we expand up to order $\varepsilon^{2(L-k)+1}$,
\[
2\sinh{(\varepsilon u_i \mu _j)}\approx \sum_{m=0}^{L-k}\frac{2(\varepsilon u_i \mu _j)^{2m+1}}{(2m+1)!}=\sum_{m=k}^{L}\frac{2(\varepsilon u_i \mu _j)^{2(L-m)+1}}{(2(L-m)+1)!}.
\]
The resulting matrix can be shown to be a product of two matrices. We use $\textbf{1}_n$ for the identity matrix of size $n$, and $\textbf{0}$ for a matrix of zeroes whose size is determined from context:
\begin{align*}
&\dt{\begin{array}{c}\left[ x_i^{\mu_j}-x_i^{-\mu_j}\right ] _{i<k}\vspace{.1cm}\\[4mm] \left[ \sum_{m=k}^L \frac{2(\varepsilon u_i\mu_j )^{2(L-m)+1}}{(2(L-m)+1)!}\right ] _{i\geq k}\end{array}}\\[4mm]
&=\det\nolimits_L{\left (\left[\begin{array}{cc}\textbf{1}_{k-1} & \textbf{0}\vspace{.1cm}\\[4mm] \textbf{0} & \left[ \frac{2(\varepsilon u_i)^{2(L-j)+1}}{(2(L-j)+1)!} \right ]_{i,j\geq k}\end{array}\right ]\left[\begin{array}{c}\left[ x_i^{\mu_j}-x_i^{-\mu_j}\right ] _{i<k}\\[4mm] \left[ \mu_j^{2(L-i)+1}\right ] _{i\geq k}\end{array}\right ]\right )}\\[4mm]
&=\frac{2^{L-k+1}\varepsilon^{(L-k+1)^2}}{\prod_{m=k}^{L}(2(L-m)+1)!}\ a^{(1)}_{(u_k,\ldots ,u_L)}\ a^{(k)}_\mu (x_1,\ldots ,x_{k-1}).
\end{align*}
This leads to the result
\begin{align}
\label{eq:aprop}
& a^{(k)}_\mu (x_1,\ldots ,x_{k-1})\\
&\qquad =\lim_{\varepsilon\rightarrow 0}\frac{\varepsilon^{-(L-k+1)^2}\prod_{m=0}^{L-k}(2m+1)!}{2^{L-k+1} a^{(1)}_{(u_k,\ldots ,u_L)}}a_\mu (x_1,\ldots, x_{k-1},\e^{\varepsilon u_k},\ldots ,\e^{\varepsilon u_L}),\nn
\end{align}
which means that proportionality factor $c_k(\varepsilon)$ is given by
\[c_k(\varepsilon)=\frac{\varepsilon^{-(L-k+1)^2}\prod_{m=0}^{L-k}(2m+1)!}{2^{L-k+1} a^{(1)}_{(u_k,\ldots ,u_L)}}.\qedhere\]
\end{proof}
We note that there is an alternative inductive proof similar to that used in \cite{KuzSkl06}, which we give in \apref{altproof}.
\begin{cor}
Because $c_k(\varepsilon)$ is independent of $\mu$, at the point $x_{k+1}=\ldots=x_L=1$ we have an alternate definition of $\chi$,
\begin{equation}
\label{eq:chitrunc}
\chi _\lambda (x_1,\ldots ,x_k,1,\ldots ,1)=\frac{ a^{(k+1)}_\mu (x_1,\ldots ,x_k)}{a^{(k+1)}_\delta (x_1,\ldots ,x_k)}.
\end{equation}
\end{cor}
This corollary allows us to directly calculate the symplectic character when some of its arguments are set to $1$.

\subsection{Factorised forms and homogeneous identities}
\label{sec:factor}
This section will list some useful identities and factorised expressions for $a_\mu$ and $\chi_\lambda$. First recall that the denominator of a Schur function is given by the Vandermonde determinant, which has the product form
\be
\label{eq:vdm}
\dt{x_i^{L-j}}=\prod_{1\leq m<n\leq L}(x_m-x_n).
\ee
The Weyl denominator formula for type $C$ gives us an analogous identity for the symplectic character \cite{FultonH91},
\be
\label{eq:ad}a_\delta (x_1,\ldots ,x_L)=\prod _{i=1}^L x_i^{i-L} \prod _{1\leq i<j\leq L}(x_i-x_j)\prod _{1\leq i\leq j\leq L}(x_i-x_j^{-1}).
\ee
We can also use \eqref{eq:vdm} to show
\begin{align}
\label{eq:am1}
a^{(1)}_\mu &=\dt{\mu _i^{2(L-j)+1}} =\prod_{i=1}^L\mu _i\prod_{1\leq i<j\leq L}(\mu _i^2-\mu _j^2).
\end{align}
This result and simple row expansion of \eqref{eq:atrunc} immediately leads to the following identity,
\be
\label{eq:am2}
a^{(2)}_\mu (z)=\sum_{k=1}^L(-1)^{k-1}\left (\prod_{i\neq k}\mu_i\prod_{\stack{1\leq i<j\leq L}{i,j\neq k}}(\mu_i^2 -\mu_j^2 )\right )(z^{\mu_k}-z^{-\mu_k}),
\ee
which will be useful later. Furthermore, from \eqref{eq:ad} and \eqref{eq:aprop}, it is easily shown that
\begin{align}
a^{(k)}_\delta (x_1,\ldots ,x_{k-1})&=\prod _{i=0}^{L-k}(2i+1)!\prod _{i=1}^{k-1}x_i^{i-L}(x_i-1)^{2(L-k+1)}\nonumber \\ &\times\prod _{1\leq i<j\leq k-1}(x_i-x_j)\prod _{1\leq i\leq j\leq k-1}(x_i-x_j^{-1}).
\label{eq:adk}
\end{align}
In particular we have
\be
\label{eq:ad1}
a^{(1)}_\delta=\prod_{i=1}^{L-1}(2i+1)!\ ,
\ee
which, along with \eqref{eq:am1}, leads to
\be
\chi_\lambda(1,\ldots,1)=\prod_{i=1}^L\frac{\mu_i}{(2i-1)!}\prod_{1\leq i<j\leq L}(\mu _i^2-\mu _j^2),
\ee
which is Weyl's dimension formula for the symplectic group \cite{Weyl39,FultonH91}. 

\section{Hamiltonians and eigenvalues}
\label{sec:Hamil}
We now construct the system of mutually commuting Hamiltonians $H_j$ which satisfy the eigenvalue equation \eqref{eq:hamil} in the symplectic case, i.e.,
\be
\label{eq:Hchi}
H_j \chi_\lambda (\mathbf{x})=h_j(\lambda )\chi_\lambda (\mathbf{x}).
\ee
Let $D_x=x\de{x}$, and recall the definition of the usual elementary symmetric function $e_j$,
\be
\label{eq:esym}
e_j(x_1,\ldots ,x_L)=\sum_{1\leq r_1<\ldots <r_j\leq L}x_{r_1}\ldots x_{r_j}.
\ee
\begin{lm}
\label{lm:Hamil}
The Hamiltonians $H_j$, for $j=1,\ldots,L$, are given by
\[ H_jf_\lambda (\mathbf x)=(a_\delta (\mathbf x) )^{-1}e_j(D_{x_1}^2,\ldots ,D_{x_L}^2)\ a_\delta (\mathbf x)f_\lambda (\mathbf x),\]
with corresponding eigenvalues
\[ h_j(\lambda )=e_j(\mu _1^2,\ldots ,\mu _L^2),\]
recalling that $\mu=\lambda+\delta$.
\end{lm}
\begin{proof}
The proof that these Hamiltonians satisfy \eqref{eq:Hchi} is equivalent to the proof of
\[ e_j(D_{x_1}^2,\ldots ,D_{x_L}^2)\ a_\mu (\mathbf{x})=e_j(\mu _1^2,\ldots ,\mu _L^2)\ a_\mu (\mathbf{x}).\]
With a slight abuse of notation, we set $v_i^j=x_i^{\mu_j}-x_i^{-\mu_j}$ and note that $D_{x_i}^2 v_i^j=\mu_j^2\ v_i^j$.
Writing $a_\mu$ as
\[ a_\mu (x_1,\ldots ,x_L)=\sum_{\sigma\in S_L}(-1)^\sigma v_1^{\sigma_1}\ldots v_L^{\sigma_L},\]
we then have
\begin{align*}
e_j(D_{x_1}^2,\ldots ,D_{x_L}^2)\ a_\mu &=\sum_{1\leq r_1<\ldots <r_j\leq L}\left (\sum_{\sigma\in S_L}D_{x_{r_1}}^2\ldots D_{x_{r_j}}^2(-1)^\sigma v_1^{\sigma_1}\ldots v_L^{\sigma_L}\right )\\
&=\sum_{\sigma\in S_L}\left(\sum_{1\leq r_1<\ldots <r_j\leq L}\mu_{\sigma_{r_1}}^2\ldots \mu_{\sigma_{r_j}}^2\right) (-1)^\sigma v_1^{\sigma_1}\ldots v_L^{\sigma_L}.
\end{align*}
Due to symmetry, the inner sum is independent of $\sigma$, and equal to $e_j(\mu _1^2,\ldots ,\mu _L^2)$. This can be taken out of the outer sum, which is equal to $a_\mu$.
\end{proof}

\section{Differential equation for $q_\lambda$}
\label{sec:diff}

In this section we describe the single variable polynomial $q_\lambda(z)$, and define the differential operator $W_i$ which satisfies \eqref{eq:Wq0}.
\begin{defn}
The polynomial $q_\lambda$ is given by
\be
\label{eq:q}
q_\lambda (z)=\frac{\chi_\lambda(z,1,\ldots,1)}{\chi_\lambda(1,\ldots,1)}.
\ee
\end{defn}
It will be useful to introduce $\phi_\mu(z_1,\ldots,z_k)$, given by
\begin{align}
\label{eq:phi}
\phi_\mu (z_1,\ldots,z_k)&=a^{(k+1)}_\mu(z_1,\ldots,z_k)/a^{(1)}_\mu,
\end{align}
and in particular
\begin{align}
\label{eq:phi1}
\phi_\mu (z)&=\sum_{j=1}^L\frac{z^{\mu_j}-z^{-\mu_j}}{\mu_j\prod_{i\neq j}(\mu _j^2-\mu _i^2)},
\end{align}
where we have used the identity \eqref{eq:am2}. Now, recalling that $\mu=\lambda+\delta$, we can write $q_\lambda (z)$ as
\be
\label{eq:qphi}
q_\lambda (z)=\phi_\mu (z)/\phi_\delta(z).
\ee
\begin{lm}
The polynomial $\phi_\mu(z)$ satisfies the differential equation
\be
\label{eq:DEphi}
\prod_{n=1}^L (D_z^2-\mu _n^2)\phi_\mu (z)=0.
\ee
\end{lm}
\begin{proof}
By definition we have
\[ D_z^2 \phi_\mu (z)=\sum_{j=1}^L \mu_j^2\, \frac{(z^{\mu_j}-z^{-\mu_j})}{\mu_j\prod_{i\neq j}(\mu _j^2-\mu _i^2)},\]
so
the $n$th term of the sum in $\phi_\mu$ is reduced to $0$ by the $n$th factor in the product.
\end{proof}

\begin{lm}
The differential equation $Wq_\lambda(z)=0$ is satisfied when $W$ is given by
\be
\label{eq:W}
W=\prod_{n=1}^L (Z^2-\mu_n^2),
\ee
where
\[ Z=D_z +\frac{(Lz^2+2Lz-2z+L)}{(z^2-1)}.\]
\end{lm}
\begin{proof}
Equation \eqref{eq:DEphi} can be written as 
\be
\label{eq:diff2}
\prod_{n=1}^L (D_z^2-\mu _n^2)\ \phi_\delta(z) q_\lambda (z)=0.
\ee
Using \eqref{eq:adk} and \eqref{eq:ad1}, we obtain
\[\phi_\delta(z) =\frac{(z+1)(z-1)^{2L-1}}{z^L(2L-1)!},\]
and applying the derivative,
\begin{align*}
D_z \phi_\delta(z) &=\frac{(z-1)^{2L-2}(Lz^2+2Lz-2z+L)}{z^L(2L-1)!}\\
&=\phi_\delta(z) \frac{(Lz^2+2Lz-2z+L)}{(z^2-1)}.
\end{align*}
Since $D_z$ obeys the same product rule as the usual derivative, we have
\begin{align*}
D_z\ \phi_\delta(z)f(z)&=\phi_\delta(z)\left (D_z+\frac{(Lz^2+2Lz-2z+L)}{(z^2-1)}\right )f(z)\\
&=\phi_\delta(z)Zf(z).
\end{align*}
Substituting this result back into \eqref{eq:diff2} we obtain
\[
0
=\prod_{n=1}^L (Z^2-\mu _n^2)\ q_\lambda (z),
\]
for $z\neq\pm 1$.
\end{proof}

\section{Inverse separating operator}
\label{sec:invS}
The inverse of the operator $\mathcal S=\mathcal A_1\ldots\mathcal A_L$ must satisfy
\[ \mathcal S^{-1}\prod _{i=1}^L q_\lambda (x_i)=\frac{\chi_\lambda (x_1,\ldots ,x_L)}{\chi_\lambda (1,\ldots ,1)},\]
for all $\lambda$.
\begin{prop}
\label{prop:SL}
$\mathcal S^{-1}$ is given by the differential operator
\[ \left(\mathcal S^{-1} f\right)(x_1,\ldots,x_L) =\\ (-1)^{L(L-1)/2} \phi_\delta (x_1,\ldots ,x_L)^{-1}\ K_L\ \prod_{i=1}^L \phi_\delta (x_i)\ f(x_1,\ldots,x_L),\]
where
\[ K_L=\dt{ D_{x_i}^{2(L-j)}}.\]
\end{prop}
\begin{proof}
Acting with the operator $\mathcal S^{-1}$ on the product of $q_\lambda(x_i)$ results, using \eqref{eq:qphi}, in $K_L$ acting on a product of $\phi_\mu(x_i)$. Taking one factor of this product into each row of the determinant $K_L$, we find 
\[ \dt{ D_{x_i}^{2(L-j)} \phi_\mu(x_i) } = \dt{ \sum _{m=1}^L\frac{\mu _m^{2(L-j)+1}(x_i^{\mu _m}-x_i^{-\mu _m})}{\mu _m^2\prod_{n\neq m}(\mu _m^2-\mu _n^2)}},\]
where we have used the explicit expression for $\phi_\mu(x_i)$ given in \eqref{eq:phi1}. This can be expressed as the product of two matrices,
\begin{align*}
\dt{\sum _{m=1}^L\frac{\mu _m^{2(L-j)+1}(x_i^{\mu _m}-x_i^{-\mu _m})}{\mu _m^2\prod_{n\neq m}(\mu _m^2-\mu _n^2)}}&=\dt{\mu_i^{2(L-j)+1}}\dt{\frac{x_i^{\mu_j}-x_i^{-\mu_j}}{\mu_j^2\prod_{n\neq j}(\mu _j^2-\mu _n^2)}}\\
&=\frac{\dt{\mu_i^{2(L-j)+1}}\dt{x_i^{\mu_j}-x_i^{-\mu_j}}}{\prod_{m=1}^L\left (\mu_m^2\prod_{n\neq m}(\mu _m^2-\mu _n^2)\right )}.
\end{align*}
The first determinant in the numerator is just $a^{(1)}_\mu$, see \eqref{eq:am1}, and the second is equal to $a_\mu(x_1,\ldots,x_L)$. Since the denominator is equal to $(-1)^{L(L-1)/2}(a^{(1)}_\mu)^2$, we finally find
\[ K_L\prod_{i=1}^L\phi_\delta (x_i) q_\lambda (x_i)=(-1)^{L(L-1)/2} \phi_\mu (x_1,\ldots ,x_L),\]
from which it immediately follows that $\mathcal S^{-1}$ satisfies the required property.
\end{proof}
We can also find the inverse of the operator $\mathcal S_k=\mathcal A_1\ldots\mathcal A_k$, which must satisfy
\[ \mathcal S_k^{-1}\prod_{i=1}^k q_\lambda (x_i)=\frac{\chi_\lambda (x_1,\ldots ,x_k,1,\ldots ,1)}{\chi_\lambda (1,\ldots ,1)}.\]
\begin{prop}
\label{prop:Sk}
$\mathcal S_k^{-1}$ is given by the differential operator
\begin{multline*} \left(\mathcal S_k^{-1}f\right)(x_1,\ldots,x_L)=\\(-1)^{(L-1)L/2+k(L+1)}\phi_\delta(x_1,\ldots ,x_k)^{-1} \dt[k]{D_{x_i}^{2(k-j)}}\prod_{i=1}^k \phi_\delta(x_i) f(x_1,\ldots,x_L).\end{multline*}
\end{prop}
$\mathcal S_k^{-1}$ is not necessary for the purposes of this paper, but will be useful for the asymptotics. We have included a proof of \propref{Sk} in \apref{Sk}.

\section{The integral operator $Q_z$}
\label{sec:Q}
In this section we will construct the operator $Q_z$ satisfying
\be
\label{eq:Qchi}\left[Q_z\chi_\lambda\right ](x_1,\ldots ,x_L)=q_\lambda (z)\,\chi_\lambda (x_1,\ldots ,x_L).
\ee
%
In order to construct $Q_z$ as an integral operator, we will need to define an appropriate domain of integration, which we will do first. The integration variables are $t_1,\ldots,t_{L-1}$, $y_1,\ldots,y_L$, and $w$. The integration variables interlace the $x$s as $x_i\leq y_i/t_i\leq x_{i+1}$, for $i=1,\ldots,L-1$.

\begin{defn}
\label{def:P}
For any Laurent series $f(t)=\sum_{m\in\mathbb{Z}} c_m t^m$ with no constant term, i.e. $c_0=0$, the domain $P$ of the integral over $t$ is defined as
\begin{align*}\int _P \frac{\dd t}{t} f(t)&:=\int _0^1\frac{\dd t}{t} \sum _{m>0} c_m t^m -\int _1^\infty\frac{\dd t}{t} \sum _{m<0} c_m t^m 
= \sum_{m\in\mathbb{Z}} \frac{c_m}{m}.
\end{align*}
For $1\leq i\leq L-1$, We also define
\[\iint_{\mathcal{D}_i}\dd Y_i=\int _P\frac{\dd t_i}{t_i}\int _{t_ix_i}^{t_ix_{i+1}}\frac{\dd y_i}{y_i},\]
and we will need two domains for the integrals over $y_L$, 
\[
x_L<\frac{y_L}{ t_1\ldots t_{L-1}}<\infty \qquad \text{and}\qquad
0<\frac{y_L}{t_1\ldots t_{L-1}}<x_1.
\]
\end{defn}

With these definitions we can now write down explicitly the operator $Q_z$ which satisfies \eqref{eq:Qchi}.
\begin{prop}
\label{prop:Q}
$Q_z$ is given by
\begin{align}
&\left[Q_z f\right ](x_1,\ldots ,x_L)=\nn\\
&\quad \frac{1}{\phi_\delta (z)a_\delta (x_1,\ldots ,x_L)}\int _1^z\frac{\dd w}{w}\int_{\mathbf D}\dd \mathbf Y\ a_\delta (y_1,\ldots ,y_L)f(y_1,\ldots ,y_L),\nn
\end{align}
where
\begin{align*}
\int_{\mathbf D} \dd\mathbf Y&=\iint _{\mathcal{D}_1}\dd Y_1\cdots\iint _{\mathcal{D}_{L-1}}\dd Y_{L-1}\left[\int_{t_1\ldots t_{L-1}x_{L}}^\infty \dd y_L\ \delta\left (y_L-wx_L\prod_{l=1}^L\frac{x_l t_l^2}{y_l}\right )\right .\\
&\quad\left .+\int _0^{t_1\ldots t_{L-1}x_1} \dd y_L\ \delta\left (y_L-\frac{x_L}{w}\prod_{l=1}^L\frac{x_l t_l^2}{y_l}\right )\right ].
\end{align*}
\end{prop}
\begin{proof}
The LHS of \eqref{eq:Qchi} becomes
\[\frac{1}{\phi_\delta (z)a_\delta (x_1,\ldots ,x_L)}\int _1^z\frac{\dd w}{w}\int_{\mathbf D} \dd \mathbf Y\ a_\mu (y_1,\ldots ,y_L),\]
while the RHS can be written as
\[\frac{ \phi_\mu (z)}{\phi_\delta(z)} \frac{a_\mu (x_1,\ldots ,x_L)}{a_\delta (x_1,\ldots ,x_L)}.\]
It is therefore enough to prove that
\begin{align}
\label{eq:Qred}
&\int _1^z\frac{\dd w}{w}\int_{\mathbf D}\dd\mathbf Y\ a_\mu (y_1,\ldots ,y_L)\stackrel{?}{=}\phi_\mu (z)\ a_\mu (x_1,\ldots ,x_L).
\end{align}
To do this, we first sketch three important steps:
\begin{enumerate}
\item[i.] We expand $a_\mu(\mathbf y)$ along the bottom row, producing
\be \sum_{r=1}^L (-1)^{L+r}(y_L^{\mu_r}-y_L^{-\mu_r})\dt[L-1]{y_i^{\mu_j}-y_i^{-\mu_j}}_{\genfrac{}{}{0pt}{}{i\neq L}{j\neq r}}.
\label{eq:aexp}
\ee
For each term in this sum, the integrals over $y_L$ can be performed easily, resulting in
\[
\label{eq:wprod}
%
y_L^{\mu_r}-y_L^{-\mu_r} \rightarrow
(w^{\mu_r}+w^{-\mu_r})\left [\left (x_L\prod _{l=1}^{L-1}\frac{x_l t_l^2}{y_l}\right )^{\mu_r}-\left (\frac1{x_L}\prod _{l=1}^{L-1}\frac{y_l}{x_l t_l^2}\right )^{\mu_r}\right ].
\]
\item[ii.] The integral over $w$ becomes elementary:
\[\int _1^z\frac{\dd w}{w}(w^{\mu_r}+w^{-\mu_r})=\frac{1}{\mu_r} \left(z^{\mu_r}-z^{-\mu_r}\right).\]
\item[iii.]
Combining \eqref{eq:aexp} and \eqref{eq:wprod} we find the further simplification
\begin{align*}
&\left [\left (x_L\prod _{l=1}^{L-1}\frac{x_l t_l^2}{y_l}\right )^{\mu_r}-\left (\frac1{x_L}\prod _{l=1}^{L-1}\frac{y_l}{x_l t_l^2}\right )^{\mu_r}\right ]\ \dt[L-1]{y_i^{\mu_j}-y_i^{-\mu_j}}_{\genfrac{}{}{0pt}{}{i\neq L}{j\neq r}}\\
&=\prod _{l=1}^{L}x_l^{\mu_r}\dt[L-1]{t_i^{2\mu_r}(y_i^{\mu_j -\mu_r}-y_i^{-\mu_j -\mu_r})}_{\genfrac{}{}{0pt}{}{i\neq L}{j\neq r}}\\
&\qquad\qquad -\prod _{l=1}^{L}x_l^{-\mu_r}\dt[L-1]{t_i^{-2\mu_r}(y_i^{\mu_j +\mu_r}-y_i^{\mu_r -\mu_j})}_{\genfrac{}{}{0pt}{}{i\neq L}{j\neq r}},
\end{align*}
after which the remaining integrals over $Y_i$ can be inserted into each row of the determinants, e.g.,
\begin{align*}
&\iint _{\mathcal{D}_1}\dd Y_1\cdots\iint _{\mathcal{D}_{L-1}}\dd Y_{L-1} \dt[L-1]{t_i^{2\mu_r}(y_i^{\mu_j -\mu_r}-y_i^{-\mu_j -\mu_r})}_{\genfrac{}{}{0pt}{}{i\neq L}{j\neq r}}\\
=&\dt[L-1]{\iint _{\mathcal{D}_{i}}\dd Y_{i}\ t_i^{2\mu_r}(y_i^{\mu_j -\mu_r}-y_i^{-\mu_j -\mu_r})}_{\genfrac{}{}{0pt}{}{i\neq L}{j\neq r}} \\
=&\dt[L-1]{\frac1{\mu_j^2 -\mu_r^2}(x_{i+1}^{\mu_j -\mu_r}-x_i^{\mu_j -\mu_r}-x_{i+1}^{-\mu_j -\mu_r}+x_i^{-\mu_j -\mu_r})}_{\genfrac{}{}{0pt}{}{i\neq L}{j\neq r}}\\
=&\frac{(-1)^{L-1}}{\prod _{j\neq r}(\mu_r^2 -\mu_j^2 )} \dt[L-1]{x_{i+1}^{-\mu_r}(x_{i+1}^{\mu_j}-x_{i+1}^{-\mu_j})-x_i^{-\mu_r}(x_i^{\mu_j}-x_i^{-\mu_j})}_{\genfrac{}{}{0pt}{}{i\neq L}{j\neq r}}.
\end{align*}
\end{enumerate}

Using these details, we finally find the following expression for the LHS of \eqref{eq:Qred}:
\begin{multline}
\label{eq:Qred2}
\sum_{r=1}^L \frac{z^{\mu_r}-z^{-\mu_r}}{\mu_r \prod _{j\neq r}(\mu_r^2 -\mu_j^2 )} \times\\
(-1)^{r-1} \left (\prod _{l=1}^{L}x_l^{\mu_r}\dt[L-1]{x_{i+1}^{-\mu_r}(x_{i+1}^{\mu_j}-x_{i+1}^{-\mu_j})-x_i^{-\mu_r}(x_i^{\mu_j}-x_i^{-\mu_j})}_{\genfrac{}{}{0pt}{}{i\neq L}{j\neq r}}\right.\\
\left. -\prod _{l=1}^{L}x_l^{-\mu_r}\dt[L-1]{x_{i+1}^{\mu_r}(x_{i+1}^{\mu_j}-x_{i+1}^{-\mu_j})-x_i^{\mu_r}(x_i^{\mu_j}-x_i^{-\mu_j})}_{\genfrac{}{}{0pt}{}{i\neq L}{j\neq r}}\right).
\end{multline}
At this point it has become clear that if $(-1)^{r-1}$ times the difference of products is independent of $r$ (so that we can factor it out of the sum), the remaining sum will be precisely equal to $\phi_\mu(z)$. Hence, it remains to show that for each $r$,
\begin{align*}
a_\mu (\mathbf x)& \stackrel{?}{=}(-1)^{r-1}\left (\prod _{l=1}^{L}x_l^{\mu_r}\dt[L-1]{x_{i+1}^{-\mu_r}(x_{i+1}^{\mu_j}-x_{i+1}^{-\mu_j})-x_i^{-\mu_r}(x_i^{\mu_j}-x_i^{-\mu_j})}_{\genfrac{}{}{0pt}{}{i\neq L}{j\neq r}}\right .\\
&\quad\left .-\prod _{l=1}^{L}x_l^{-\mu_r}\dt[L-1]{x_{i+1}^{\mu_r}(x_{i+1}^{\mu_j}-x_{i+1}^{-\mu_j})-x_i^{\mu_r}(x_i^{\mu_j}-x_i^{-\mu_j})}_{\genfrac{}{}{0pt}{}{i\neq L}{j\neq r}}\right ).
\end{align*}

This can be achieved by increasing the size of the matrices in the determinants by one row and one column, at $i=1$ and $j=r$, cancelling the factor of $(-1)^{r-1}$. The new column has entries of $0$ except for at $i=1$, which is 1, and the new row has entries $x_1^{-\mu_r}(x_1^{\mu_j}-x_1^{-\mu_j})$ for the first determinant and $x_1^{\mu_r}(x_1^{\mu_j}-x_1^{-\mu_j})$ for the second. We then use row reduction, adding the $1$st row to the $2$nd row, then the $2$nd to the $3$rd, and so on. For ease of display, we place the additional column on the end, producing an extra factor of $(-1)^{L-r}$ which will be cancelled at the very end when we permute the column back to $j=r$. Thus we obtain
\begin{multline*}
(-1)^{L-r} \left( \prod_{l=1}^L x_l^{\mu_r}\det\nolimits_L\Bigg[{\begin{array}{ll}
\Big[x_i^{-\mu_r}(x_i^{\mu_j}-x_i^{-\mu_j})\Big ]_{j\neq r} & \Big[1\Big ] 
\end{array}}\Bigg ] \right.\\
\left.-
\prod_{l=1}^L x_l^{-\mu_r}\det\nolimits_L\Bigg[{\begin{array}{ll}
\Big[x_i^{\mu_r}(x_i^{\mu_j}-x_i^{-\mu_j})\Big ]_{j\neq r} & \Big[1\Big ]
\end{array}}\Bigg ]\right),
\end{multline*}
where $[1]$ denotes a column full of 1's. Each factor of the products outside the determinants is then inserted into its respective row, resulting in
\[
=(-1)^{L-r} \left(\det\nolimits_L\Bigg[{\begin{array}{ll}
\Big[x_i^{\mu_j}-x_i^{-\mu_j}\Big ]_{j\neq r} & \Big[x_i^{\mu_r}\Big ]
\end{array}}\Bigg ]-\det\nolimits_L\Bigg[{\begin{array}{lll}
\Big[x_i^{\mu_j}-x_i^{-\mu_j}\Big ]_{j\neq r} & \Big[x_i^{-\mu_r}\Big ]
\end{array}}\Bigg ]\right).
\]
These two determinants can be combined by simply performing the subtraction in the last column. Permuting this column back to $j=r$, this finally results in
\[\dt{x_i^{\mu_j}-x_i^{-\mu_j}}=a_\mu (x_1,\ldots ,x_L).\qedhere\]
\end{proof}

\section{The integral operator $A_k$}
\label{sec:A}
The next and final step is to obtain the operator $\mathcal{A}_k$, which satisfies
\begin{align}
\label{eq:QA}\big(\rho _{k-1}Q_{z}f(\mathbf{x})\big)\big|_{z=x_k} =\mathcal A_k\rho _k f(\mathbf{x}),
\end{align}
where
\[\rho _kf(x_1,\ldots ,x_L)=f(x_1,\ldots ,x_k,1,\ldots ,1).\]
\begin{prop}
\label{prop:A}
Relation \eqref{eq:QA} is satisfied by
\begin{align*}
\mathcal{A}_k&=\lim_{\varepsilon\rightarrow 0}\frac{1}{2\varepsilon} \frac1{\phi_\delta (x_k) a^{(k)}_\delta (x_1,\ldots ,x_{k-1})}\int _1^{x_k}\frac{\dd w}{w}\int_{\mathbf D^{(\varepsilon)}} \dd \mathbf Y^{(\varepsilon)}\  a^{(k+1)}_\delta(y_1,\ldots ,y_k),
\end{align*}
where
\begin{align*} \int_{\mathbf D^{(\varepsilon)}} \dd\mathbf Y^{(\varepsilon)}&=\iint _{\mathcal D_1}\dd Y_1\cdots\iint _{\mathcal D_{k-1}}\dd Y_{k-1}\left[\int _0^{t_1\ldots t_{k-1}x_1} \dd y_k\ \delta\left (y_k-\frac{\e^\varepsilon}{w}\prod_{l=1}^{k-1}\frac{t_l^2 x_l}{y_l}\right )\right .\\
&\left .+\int_{t_1\ldots t_{k-1}\e^\varepsilon}^\infty \dd y_k\ \delta\left (y_k-w \e^\varepsilon\prod_{l=1}^{k-1}\frac{t _l^2 x_l}{y_l}\right )\right ],
\end{align*}
and
\[\iint_{\mathcal D_{k-1}}\dd Y_{k-1}=\int _P\frac{\dd t_{k-1}}{ t_{k-1}}\int_{ t_{k-1}x_{k-1}}^{ t_{k-1}\e^\varepsilon}\frac{\dd y_{k-1}}{y_{k-1}}.\]
We remind the reader that the domain of integration $P$ is defined in \defref{P}.
\end{prop}
\begin{proof}
For simplicity, assume that $f=\chi_\lambda$ (we will generalise this later), so $a^{(k+1)}_\delta\rho_kf=a^{(k+1)}_\mu$. We want to prove that
$\mathcal{A}_k \rho_k\chi_\lambda(\mathbf{x})=q_\lambda(x_k)\rho_{k-1}\chi_\lambda(\mathbf{x})$, i.e.
\begin{align}
\label{eq:Aam}
\lim_{\varepsilon\rightarrow 0}\frac1{2\varepsilon}\int _1^{x_k}\frac{\dd w}{w}\int_{\mathbf D^{(\varepsilon)}} \dd \mathbf Y^{(\varepsilon)}\  a^{(k+1)}_\mu(y_1,\ldots ,y_k)=\phi_\mu(x_k)a^{(k)}_\mu(x_1,\ldots,x_{k-1}).
\end{align}

This proof is very similar to that of \propref{Q}, but with a few added subtleties. The differences in the proof are outlined here, and we give the details in \apref{Ak}.

As in step i. of the proof of \propref{Q}, the determinant $a^{(k+1)}_\mu(x_1,\ldots,x_k)$ is expanded along a row, but this time along the $k$th row instead of the last. Also, the integrals over $y_k$ take the place of that over $y_L$ in this step. The simplification in step iii. is also very similar. The products outside the determinant now run from $1$ to $k-1$, and contain an extra factor of $\e^{\pm\varepsilon\mu_r}$. The $k-1$ integrals over $Y_1,\ldots Y_{k-1}$ get inserted into the first $k-1$ rows of each determinant.

After these integrals are performed, the sizes of the matrices are increased as before. This time we choose the entries of the new row proportional to $\sinh \varepsilon\mu_j$ so that the matrices simplify after row reduction. After this, the prefactors in front of the two determinants can be cancelled, and the determinants can be combined.

The final step is to approximate to first order in $\varepsilon$, and then perform row reduction on the rows from $k$ to $L$, in order to remove the dependence on $\mu_r$ from each row. The result is a factor $a^{(k)}_\mu(x_1,\ldots,x_{k-1})$, which can be factored out of the remaining sum over $r$. This last sum is simply equal to $\phi_\mu(x_k)$, and the proof is complete for $f=\chi_\lambda$. Since the $\chi_\lambda$ form a linear basis, it follows that \eqref{eq:QA} holds for any polynomial in $\mathbb C[y_1^\pm,\ldots ,y_k^\pm ]^{\mathcal W_k}$.

\end{proof}

\section{Conclusion}
\label{sec:conc}
We have set out the $Q$-operator method of separation of variables for the symplectic character according to the steps laid out in \secref{method}. The separating operator is constructed from a chain of operators $\mathcal A_k$, each of which splits off the dependence of a single variable from $\chi_\lambda$. The factorised Hamiltonian is constructed as described in \secref{fachamil} from the differential equation satisfied by $q_\lambda$.

The separation of variables for the symplectic character follows along similar lines as that of the type $A$ Schur polynomial. However, some of the intermediate steps are technically much more involved. In particular, the $Q$ and $\mathcal A$-operators contain a double integration in each variable, whereas the corresponding operators for the Schur function contain only one. This is related to the fact that the Hamiltonian for the symplectic character contains double derivatives.

An obvious extension of this work is to generalise the method to Jack polynomials of type $BC$, or even Koornwinder polynomials. It is not expected that the inverse separating operator will be easy to construct in the more general case, as it is still an unsolved problem in the case of type $A$ Jack polynomials. However, another method for reversing the SoV process was used in \cite{KuzMS03}, called the lifting operator, and it is hoped that an analog of this might prove useful for the more general $BC$ polynomials.

The problem of asymptotics for this polynomial can now be regarded as a problem of asymptotics for the separated polynomial, $q_\lambda$. As mentioned earlier, we are particularly interested in the asymptotic limit of $\chi_\lambda$ when all but $k$ variables are set to 1, and the operator $\mathcal S_k$ of \propref{Sk} and its inverse are useful here. The differential equation satisfied by $q_\lambda$, given in \secref{diff}, will be a useful tool in determining the asymptotics of $q_\lambda$.


\section*{Acknowledgment}
\label{sec:ack}
We thank Vladimir Mangazeev for suggesting the SoV approach to us, as well as for interesting discussions. We are indebted to the Australian Research Council (ARC) for financial support. AP is grateful for the hospitality of the Rudolph Peierls Centre for Theoretical Physics at the University of Oxford where part of this work was undertaken, and would also like to thank Robert Weston for useful discussions.


\appendix
\section{Inductive proof of \lmref{kprop}}
\label{ap:altproof}
\begin{proof}
We take $\lim_{\varepsilon\rightarrow 0}a^{(k+1)}_\mu(x_1,\ldots,x_{k-1},\e^\varepsilon)$. In the $k$th row of the determinant, the $j$th element is
\[
\e^{\varepsilon\mu _j}-\e^{-\varepsilon\mu _j}=2\sinh(\varepsilon\mu _j)=2\sum_{n=0}^\infty \frac{(\varepsilon\mu _j)^{2n+1}}{(2n+1)!}.
\]
In the determinant \eqref{eq:atrunc}, we use row reduction with the rows below to remove the terms up to $n=L-(k+1)$. The remainder of the series is
\[ 2\frac{(\varepsilon\mu _j)^{2(L-k)+1}}{(2(L-k)+1)!}+ {\rm h.o.t.},\]
so $ a^{(k)}_\mu(x_1,\ldots,x_{k-1})=\lim_{\varepsilon\rightarrow 0}\frac{(2(L-k)+1)!}{2\varepsilon^{2(L-k)+1}} a^{(k+1)}_\mu(x_1,\ldots,x_{k-1},\e^\varepsilon)$.
\end{proof}

\section{Proof of \propref{Sk}}
\label{ap:Sk}
\begin{proof}
The proof of the required property of $\mathcal S_k^{-1}$ is equivalent to the proof of
\be
\label{eq:redSk}
\dt[k]{D_{x_i}^{2(k-j)}}\prod_{i=1}^k \phi_\mu (x_i)=(-1)^{\frac{(L-1)L}{2}+k(L+1)}\phi_\mu(x_1,\ldots ,x_k).
\ee
On the LHS, we insert one factor $\phi_m(x_i)$ into row $i$ of the determinant, so that each element of the matrix becomes
\[ D_{x_i}^{2(k-j)}\phi_\mu (x_i)=\sum_{m=1}^L\mu _m^{2(k-j)}\ \frac{(x_i^{\mu_m}-x_i^{-\mu_m})}{\mu _m\prod_{n\neq m}(\mu _m^2-\mu _n^2)}.\]
Using the proof of \propref{SL} as a guide, we rewrite the RHS as
\[ (-1)^{\frac{(L-1)L}{2}+k(L+1)}\frac{ a_\mu^{(1)} a_\mu^{(k+1)}(x_1,\ldots ,x_k)}{\left ( a_\mu^{(1)}\right )^2},\]
and use the product formula of the denominator while using the determinant formula for the numerator, to get
\begin{align}
\label{eq:messdetk}&(-1)^{\frac{(L-1)L}{2}+k(L+1)}\frac{\dt{\begin{array}{c}\left[x_i^{\mu_j}-x_i^{-\mu_j}\right ]_{i\leq k}\\[4mm] \left[\mu _j^{2(L-i)+1}\right ]_{i>k}\end{array}}\dt{\mu _i^{2(L-j)+1}}}{\left (\prod_{m}\mu _m\prod_{n>m}(\mu _m^2-\mu _n^2)\right )^2}\nn\\[4mm]
&=(-1)^{k(L+1)}\dt{\begin{array}{c}\left[x_i^{\mu_j}-x_i^{-\mu_j}\right ]_{i\leq k}\\[4mm] \left[\mu _j^{2(L-i)+1}\right ]_{i>k}\end{array}}\dt{\frac{\mu _i^{2(L-j)+1}}{\mu _i^2\prod_{n\neq i}(\mu _i^2-\mu _n^2)}}\nn\\[4mm]
&=(-1)^{k(L+1)}\dt{\begin{array}{c}\left[\sum_{m=1}^L \frac{\mu _m^{2(L-j)}(x_i^{\mu_m}-x_i^{-\mu_m})}{\mu _m\prod_{n\neq m}(\mu _m^2-\mu _n^2)}\right ]_{i\leq k}\\[4mm] \left[\sum_{m=1}^L \frac{\mu _m^{4L-2(i+j)+2}}{\mu _m^2\prod_{n\neq m}(\mu _m^2-\mu _n^2)}\right ]_{i>k}\end{array}}.
\end{align}
The elements in rows $k$ to $L$ can be written, with $\eta=i+j$, as
\begin{align*}
&\frac1{\prod_{r<s}(\mu _r^2-\mu _s^2)}\sum_{m=1}^L (-1)^{m-1}\mu _m^{4L-2\eta}\prod_{\stack{1\leq r<s\leq L}{r,s\neq m}}(\mu _r^2-\mu _s^2)\\
&=\frac{\dt{\begin{array}{cc}\left[\mu_i^{4L-2\eta}\right ]_{j=1} & \left[\mu_i^{2(L-j)}\right ]_{j\geq 2}\end{array}}}{\dt{\mu _i^{2(L-j)}}},
\end{align*}
which is $0$ when $\eta>L+1$, and $1$ when $\eta=L+1$. This means that \eqref{eq:messdetk} can be expressed as
\[ (-1)^{k(L+1)}\dt{\begin{array}{cc}\left[*\right ]_{\stack{i\leq k}{j\leq L-k}} &\left[\sum_{m=1}\frac{\mu_m^{2(L-j)}(x_i^{\mu_m}-x_i^{-\mu_m})}{\mu_m\prod_{n\neq m}(\mu _m^2-\mu _n^2)}\right ]_{\stack{i\leq k}{j\geq L-k+1}}\\[6mm] \text{\bf A} &\textbf{0}_{\stack{i\geq k+1}{j\geq L-k+1}}\end{array}},\]
where `$*$' is an entry which does not contribute to the determinant, and \textbf{A} is a matrix with $1$s on the backwards diagonal, $0$s below and `$*$'s above. Then the expression can be reduced to
\begin{align*}
&\dt[k]{\sum_{m=1}^L \frac{\mu _m^{2(k-j)}(x_i^{\mu_m}-x_i^{-\mu_m})}{\mu _m\prod_{n\neq m}(\mu _m^2-\mu _n^2)}},
\end{align*}
which is exactly the LHS of \eqref{eq:redSk}.
\end{proof}

\section{Proof of \propref{A}}
\label{ap:Ak}
\begin{proof}
The proof of equation \eqref{eq:Aam} goes along similar lines as for \propref{Q}:
\begin{itemize}
\item[i.]
The first step is to expand the determinant $a^{(k+1)}_\mu$ over the $k$th row, giving
\be
\label{eq:sumA} \sum_{r=1}^L (-1)^{k+r}(y_k^{\mu_r}-y_k^{-\mu_r})\dt[L-1]{\begin{array}{c}
\left[y_i^{\mu_j}-y_i^{-\mu_j}\right ]_{i<k}\\[4mm]
\left[\mu_j^{2(L-i)+1}\right ]_{i>k}
\end{array}}_{j\neq r}.\ee
We can perform the integrals over $y_k$ in \eqref{eq:Aam}, and, as before, the dependence on $w$ factors out. 
\item[ii.] The $w$ integral can be evaluated.
\item[iii.] The remaining factor is combined with the determinant,
\begin{align*}
&\left[\left (\e^\varepsilon\prod_{l=1}^{k-1}\frac{t _l^2 x_l}{y_l}\right )^{\mu_r}-\left (\e^{-\varepsilon}\prod_{l=1}^{k-1}\frac{y_l}{t_l^2 x_l}\right )^{\mu_r}\right ]\dt[L-1]{\begin{array}{c}
\left[y_i^{\mu_j}-y_i^{-\mu_j}\right ]_{i<k}\\[4mm]
\left[\mu_j^{2(L-i)-1}\right ]_{k\leq i<L}
\end{array}}_{j\neq r}\\
&=\e^{\varepsilon\mu_r}\prod_{l=1}^{k-1}x_l^{\mu_r}\dt[L-1]{\begin{array}{c}
\left[t_i^{2\mu_r}(y_i^{\mu_j -\mu_r}-y_i^{-\mu_j -\mu_r})\right ]_{i<k}\\[4mm]
\left[\mu_j^{2(L-i)-1}\right ]_{k\leq i<L}
\end{array}}_{j\neq r}\\
&\quad -\e^{-\varepsilon\mu_r}\prod_{l=1}^{k-1}x_l^{-\mu_r}\dt[L-1]{\begin{array}{c}
\left[t_i^{-2\mu_r}(y_i^{\mu_j +\mu_r}-y_i^{-\mu_j +\mu_r})\right ]_{i<k}\\[4mm]
\left[\mu_j^{2(L-i)-1}\right ]_{k\leq i<L}
\end{array}}_{j\neq r}.
\end{align*}
Now, the $k-1$ integrals over $Y_i$ in \eqref{eq:Aam} can be inserted into the first $k-1$ rows of the determinants, and evaluated as before. The two determinants then become
\begin{align*}
&\e^{\pm\varepsilon\mu_r} \prod_{l=1}^{k-1}x_l^{\pm\mu_r}\dt[L-1]{\begin{array}{c}
\left[\frac1{\mu_j^2 -\mu_r^2}(x_{i+1}^{\mp\mu_r}(x_{i+1}^{\mu_j}-x_{i+1}^{-\mu_j})-x_i^{\mp\mu_r}(x_i^{\mu_j}-x_i^{-\mu_j}))\right ]_{i<k-1}\\[4mm]
\frac1{\mu_j^2 -\mu_r^2}(2\e^{\mp\varepsilon\mu_r}\sinh{\varepsilon\mu_j}-x_{k-1}^{\mp\mu_r}(x_{k-1}^{\mu_j}-x_{k-1}^{-\mu_j}))\\[4mm]
\left[\mu_j^{2(L-i)-1}\right ]_{k-1<i<L}
\end{array}}_{j\neq r}\\[4mm]
&=\frac{(-1)^{L-1}\e^{\pm\varepsilon\mu_r} \prod_{l=1}^{k-1}x_l^{\pm\mu_r}}{\prod_{j\neq r}(\mu_r^2 -\mu_j^2 )}\dt[L-1]{\begin{array}{c}
\left[x_{i+1}^{\mp\mu_r}(x_{i+1}^{\mu_j}-x_{i+1}^{-\mu_j})-x_i^{\mp\mu_r}(x_i^{\mu_j}-x_i^{-\mu_j})\right ]_{i<k-1}\\[4mm]
2\e^{\mp\varepsilon\mu_r}\sinh{\varepsilon\mu_j}-x_{k-1}^{\mp\mu_r}(x_{k-1}^{\mu_j}-x_{k-1}^{-\mu_j})\\[4mm]
\left[(\mu_j^2 -\mu_r^2 )\mu_j^{2(L-i)-1}\right ]_{k-1<i<L}
\end{array}}_{j\neq r}.
\end{align*}
Again, as in \eqref{eq:Qred2}, we will extract the ingredients needed for $\phi_\mu(x_k)$ and show that the rest is independent of $r$ so that it can be factored out of the sum in \eqref{eq:sumA}.
\end{itemize}

As in the previous section, we increase the size of each matrix by adding a row at $i=L$ and a column at $j=r$, introducing a factor of $(-1)^{L+r}$. The new column has entries equal to $0$ except at $i=L$, which equals $1$, and the new row has entries of $2\e^{\mp\varepsilon\mu_r}\sinh{\varepsilon\mu_j}$. We depict column $r$ at the end:
\[(-1)^{r+1}\e^{\pm\varepsilon\mu_r} \prod_{l=1}^{k-1}x_l^{\pm\mu_r}
\dt{\begin{array}{cc}
\left[(x_{i+1}^{\mp\mu_r}(x_{i+1}^{\mu_j}-x_{i+1}^{-\mu_j})-x_i^{\mp\mu_r}(x_i^{\mu_j}-x_i^{-\mu_j}))\right ]_{i<k-1, j\neq r} & \mathbf{0}\\[4mm]
\left[2\e^{\mp\varepsilon\mu_r}\sinh{\varepsilon\mu_j}-x_{k-1}^{\mp\mu_r}(x_{k-1}^{\mu_j}-x_{k-1}^{-\mu_j})\right ]_{j\neq r} & 0\\[4mm]
\left[(\mu _j^2-\mu _r^2)\mu _j^{2(L-i)-1}\right ]_{k-1<i<L, j\neq r} & \mathbf{0}\\[4mm]
\left[2\e^{\mp\varepsilon\mu_r}\sinh{\varepsilon\mu_j}\right ]_{j\neq r} & 1
\end{array}}\\[4mm]
\]

Now we use row reduction, subtracting the $L$th row from the $k$th and then adding the $(k-1)$th to the $k$th, the $(k-2)$th to the $(k-1)$th, etc. We then multiply the first $k$ rows by $-x_i^{\pm\mu_r}$, and the $L$th by $\e^{\pm\varepsilon\mu_r}$, and get
\[
(-1)^{k-1}\dt{\begin{array}{cc}
\left[(x_i^{\mu_j}-x_i^{-\mu_j}))\right ]_{i<k, j\neq r} & \left[x_i^{\pm\mu_r}\right]_{i<k}\\[4mm]
\left[(\mu _j^2-\mu _r^2)\mu _j^{2(L-i)-1}\right ]_{k\leq i<L, j\neq r} & \mathbf{0}\\[4mm]
\left[2\sinh{\varepsilon\mu_j}\right ]_{j\neq r} & \e^{\pm\varepsilon\mu_r}
\end{array}}.
\]
At this point we are able to combine the two determinants by subtracting the $r$th column of the second from the $r$th column of the first, obtaining
\[
\dt{\begin{array}{cc}
\left[(x_i^{\mu_j}-x_i^{-\mu_j}))\right ]_{i<k, j\neq r} & \left[(x_i^{\mu_r}-x_i^{-\mu_r}))\right ]_{i<k}\\[4mm]
\left[(\mu _j^2-\mu _r^2)\mu _j^{2(L-i)-1}\right ]_{k\leq i<L, j\neq r} & \mathbf{0}\\[4mm]
\left[2\sinh{\varepsilon\mu_j}\right ]_{j\neq r} & 2\sinh{\varepsilon\mu_r}
\end{array}}.
\]
Now we take the limit as $\varepsilon\rightarrow 0$ and approximate to first order, resulting in $2\varepsilon\mu_j$ in the bottom row. We can factor out $2\varepsilon$, and then use row reduction once again: To each row in turn from $i=L-1$ to $i=k$ we add $\mu_r^2$ times the row below, and we are finally left with
\begin{align*}
2\varepsilon\dt{\begin{array}{c}
\left[(x_i^{\mu_j}-x_i^{-\mu_j}))\right ]_{i<k}\\[4mm]
\left[\mu _j^{2(L-i)+1}\right ]_{i\geq k}
\end{array}}=2\varepsilon  a_\mu^{(k)}(x_1,\ldots ,x_{k-1}).
\end{align*}
This then factors out of the sum in \eqref{eq:sumA}, which is equal to $\phi_\mu(x_k)$. Putting everything together, the factors of $(-1)$ cancel out, as does the factor of $2\varepsilon$, and we are left with
\[
a^{(k)}_\mu (x_1,\ldots ,x_{k-1})\phi_\mu (x_k),
\]
which is the RHS of \eqref{eq:Aam}.
\end{proof}

\end{document}